\documentclass[runningheads]{llncs}
\usepackage[T1]{fontenc}

\usepackage[dvipsnames]{xcolor}
\usepackage{hyperref}
\usepackage{amsmath,amscd}
\usepackage{amssymb}
\usepackage{upref}
\usepackage{multirow}
\usepackage{multicol}
\usepackage{url}
\usepackage[all]{xy}

\usepackage[whileod]{clr-alg}
\colorlet{shadecolor}{gray!20} 

\usepackage{framed}

\usepackage{color}
\usepackage{tikz}
\usetikzlibrary{automata} 
\usetikzlibrary{positioning} 
\usetikzlibrary{arrows} 
\tikzset{node distance=1.5cm, 
            every state/.style={ 
                  semithick,
                  fill=gray!10},
            initial text={}, 
            double distance=2pt, 
            every edge/.style={ 
                   draw,
                   ->,>=stealth', 
                   auto,
            semithick}}
\usepackage{lineno}

\newcommand\A{\mathcal{A}}
\newcommand\AP{\mathcal{A'}}

\newcommand{\Suff}{\mbox{\it Suff\/}}
\newcommand{\DAWG}{\mbox{\it DAWG\/}}
\newcommand{\out}{\mbox{\it Out\/}}

\DeclareMathOperator{\Mark}{mark}
\DeclareMathOperator{\size}{size}
\DeclareMathOperator{\Fact}{Fact}

\usepackage{xspace}

\usepackage{calc}
\newcounter{hours}\newcounter{minutes}

\definecolor{lime}{HTML}{A6CE39}
\DeclareRobustCommand{\orcidicon}{%
	\begin{tikzpicture}
	\draw[lime, fill=lime] (0,0)
	circle [radius=0.16]
	node[white] {{\fontfamily{qag}\selectfont \tiny ID}};
	\draw[white, fill=white] (-0.0625,0.095)
	circle [radius=0.007];
	\end{tikzpicture}
	\hspace{-2mm}
}
\foreach \x in {A, ..., Z}{%
 \expandafter\xdef\csname orcid\x\endcsname{\noexpand%
 \href{https://orcid.org/\csname orcidauthor\x\endcsname}{\noexpand\orcidicon}}
}

\begin{document}
\title{Fast detection of specific fragments against a set of sequences}
%
%
\author{Marie-Pierre B\'eal\inst{1}\orcidA{} \and
Maxime Crochemore\inst{1}\orcidB{}}
\authorrunning{Marie-Pierre B\'eal et al.}
%
\institute{Univ. Gustave Eiffel, CNRS, LIGM, F-77454
  Marne-la-Vall\'ee, France
\email{\{marie-pierre.beal,maxime.crochemore\}@univ-eiffel.fr}
}
\maketitle              

\begin{abstract}
We design alignment-free techniques for comparing a sequence or word, called a target, against a set of words, called a reference. A target-specific factor of a target $T$ against a reference $R$ is a factor $w$ of a word in $T$ which is not a factor of a word of $R$ and such that any proper factor of $w$ is a factor of a word of $R$.
We first address the computation of the set of target-specific factors of a target $T$ against a reference $R$, where $T$ and $R$ are finite sets of sequences. The result is the construction of an automaton accepting the set of all considered target-specific factors. The construction algorithm runs in linear time according to the size of $T\cup R$.
The second result consists of the design of an algorithm to compute
all the occurrences in a single sequence $T$ of its target-specific
factors against a reference $R$. The algorithm runs in real-time on
the target sequence, independently of the number of occurrences of
target-specific factors.
\end{abstract}

\section{Introduction}
The goal of this article is to design an alignment-free technique for comparing a sequence or word, called a target, against a set of words, called a reference. 

The motivation comes from the analysis of genomic sequences as done for example by Khorsand et al. in \cite{BonizzoniEtAl2021} in which authors introduce the notion of sample-specific strings. To avoid alignments but to extract interesting elements that differentiate the target from the reference, the chosen specific fragments are minimal forbidden factors, also called minimal absent factors. Target-specific words are factors of the target that are minimal forbidden factors of the reference. These types of factors have already been applied to compare efficiently sequences (see for example \cite{DBLP:journals/iandc/Charalampopoulos18} and references therein), to build phylogenies of biological molecular sequences using a distance based on absent words (see \cite{DBLP:journals/tcs/ChairungseeC12,DBLP:journals/corr/abs-2105-14990},...), to discover remarkable patterns in some genomic sequences (see for example \cite{DBLP:journals/bioinformatics/SilvaPCPF15}) and to improve pattern matching methods (see \cite{DBLP:journals/iandc/CrochemoreHKMPR20}...), to quote only a few applications. In bioinformatics target-specific words act as signatures for newly sequenced biological molecules and help find their characteristics.

The notion of minimal absent factors was introduced by Mignosi et al. \cite{DBLP:conf/birthday/MignosiRS99} (see also \cite{DBLP:journals/aam/BealMRS00}) in relation to combinatorial aspects of some sequences. It has then been extended to regular languages in \cite{DBLP:journals/fuin/BealCM03}, which obviously applies to a finite set of (finite) sequences.
The first linear-time computation is described in \cite{DBLP:journals/ipl/CrochemoreMR98} (see also \cite{CHL07cup}) and, due to the important role of the notion, the efficient computation of minimal forbidden factors has attracted quite a lot of works (see for example \cite{DBLP:journals/bmcbi/PinhoFGR09} and references therein).

In the article, we continue exploring the approach of target-specific words as done in \cite{BonizzoniEtAl2021} by introducing new other algorithmic techniques to detect them. See also the more general view on the usefulness of formal languages to analyze several genomes using pangenomics graphs by Bonizzoni et al. in \cite{BonizzoniEtAl2022}.  

\paragraph{The results.}
First, we address the computation of the set of target-specific factors of  a target $T$ against a reference $R$, where $T$ and $R$ are finite sets of sequences. The result is the construction of an automaton accepting the set of all considered target-specific factors. The construction algorithm runs in linear time according to the size of $T\cup R$.

The second result consists of the design of an algorithm to compute all the occurrences in a single sequence $T$ of its target-specific factors against a reference $R$. The algorithm runs in real-time on the target sequence, independently of the number of occurrences of target-specific factors, after a standard processing of the reference. This improves on the result in \cite{BonizzoniEtAl2021}, where the running time of the main algorithm depends on the number of occurrences of sought factors.

The design of both algorithms uses the notion of suffix links that are used for building efficiently indexing data structures, like suffix trees (see \cite{DBLP:books/cu/Gusfield1997}) and DAWGs also called suffix automata (see \cite{BlumerEtAl1984,CHL07cup}). The links can also be simulated with suffix arrays \cite{DBLP:journals/siamcomp/ManberM93} and their implementations, for example, the FM-index \cite{DBLP:conf/focs/FerraginaM00}. The algorithm in \cite{BonizzoniEtAl2021} uses the FMD index by Li \cite{10.1093/bioinformatics/bts280}. All these data structures can accommodate the sequences and their reverse complements.

\paragraph{Definitions.}
Let $A$ be a finite alphabet and $A^*$ be the set of the finite words
 drawn from the alphabet $A$, including the empty word $\varepsilon$.
A \emph{factor} of a word $u \in A^*$ is a word $v \in A^*$ that satisfies $u = wvt$ for some words $w, t \in A^*$.
A \emph{proper factor} of a word $u$ is a factor distinct from the whole word.
If $P$ is a set of words, we denote by $\Fact(P)$ the set of factors of words in $P$, and, if $P$ is finite, $\size(P)$ denotes the sum of lengths of the words in $P$.

A \emph{minimal forbidden word} (also called a minimal absent word) for a given set of words $L \subseteq A^*$ with respect to a given alphabet $B$ containing $A$ is a word of $B^*$ that does not belong to $L$ but that all proper factors do.

Let $R, T$ be two sets of finite words. A \emph{$T$-specific word
  with respect to $R$} is a word $u$ for which:
  $u$ is a factor of a word of $T$,
  $u$ is not a factor of a word in $R$
  and any proper factor of $u$ is a factor of a word in $R$.
The set $R$ is called the \emph{reference} and $T$ the \emph{target} of the problem. 

Note that a word is a $T$-specific word with respect to $R$ if and only if it is a minimal forbidden word of $\Fact(R)$ with respect to the alphabet of letters occurring in $R \cup T$ and is also in $\Fact(T)$. As a consequence, the set of $T$-specific words with respect to $R$ is both prefix-free and suffix-free. 

It follows from the definition that the set $S$ of $T$-specific words with respect to $R$ is:
\[
A\Fact(R) \cap \Fact(R)A \cap (A^* - \Fact(R)) \cap \Fact(T),
\]
where $A$ is the alphabet of letters of words $R$ and $T$.
It is thus a regular set when $R$ and $T$ are regular, in particular when $R$ and $T$ are finite.

A \emph{finite deterministic automaton} is denoted by $\A = (Q, A, i, F, \delta)$ where $A$ is a finite alphabet, $Q$ is a finite set of states, $i \in Q$ is the unique initial state, $F \subseteq Q$ is the set of final states and $\delta$ is the partial function from $Q \times A$ to $Q$ representing the transitions of the automaton. The partial function $\delta$ extends to $Q \times A^*$ and a word $u$ is accepted by $\A$ if and only if $\delta(i, u)$ is defined and belongs to $F$.

\section{Background: directed acyclic word graph}

In this section, we recall the definition and the construction of the directed acyclic word graph of a finite set of words. This description already appears in \cite{DBLP:journals/fuin/BealCM03}.

Let~$P=\{x_1, x_2, \ldots, x_r\}$ be a finite set of words of size~$r$.
A linear-time construction of a deterministic finite state automaton recognizing~$\Fact(P)$ has been obtained
by Blumer~\emph{et al.\ }in~\cite{BlumerEtAl1984}, \cite{Blumer87complete}, see also~\cite{NavarroRaffinot2002}.
Their construction is an extension of the
well-known incremental construction of the suffix automaton of a single word
(see for instance \cite{Cro86tcs,CHL07cup}).  The words are added one by one
to the automaton. In the sequel, we call this algorithm the {\sc Dawg}
algorithm since it outputs a deterministic automaton called a \emph{directed
acyclic word graph}. Let us denote by~$\DAWG(P)=(Q,A,i,Q,\delta)$ this
automaton. Let $\Suff(v)$ denote the set of suffixes of a word~$v$ and
$\Suff(P)$ the union of all $\Suff(v)$ for $v \in P$. The states
of~$\DAWG(P)$ are the equivalence classes of the right invariant
equivalence~$\equiv_{\Suff(P)}$ defined as follows. If $u, v \in \Fact(P)$,
\[ u \equiv_{\Suff(P)} v \mbox{  iff  } \forall i 1 \leq i \leq r \mbox{ and } \ \ u^{-1}\Suff(x_i)= v^{-1}\Suff(x_i).
\]
and there is a transition labeled by~$a$ from the class of a
word~$u$ to the class of~$ua$. The automaton~$\DAWG(P)$ has a
unique initial state, which is the class of the empty word, and
all its states are final. Note that the syntactic congruence $\sim$
defining the minimal automaton of the language is
\[ u \sim v \mbox{  iff  } \bigcup_{i=1}^{r} u^{-1}\Suff(x_i)= \bigcup_{i=1}^{r} v^{-1}\Suff(x_i)
\]
and is not the same as the above equivalence. In other words,
$\DAWG(P)$ is not always a minimal automaton.

The construction of~$\DAWG(P)$ is performed in time~$O(\size(P)\times \log
|A|)$. A time complexity of~$O(\size(P))$ can be obtained with an
implementation of automata with sparse matrices (see \cite{CHL07cup}).

\begin{example} The deterministic acyclic word graph obtained with
the {\sc Dawg} algorithm from $P= \{abbab, abaab\}$ is displayed
in Figure~\ref{figure.1} where dashed edges represent the suffix links.
Note that this deterministic automaton
is not minimal since states $3$ and $7$, $5$ and $9$, and $6$ and
$10$ can be merged pairwise.
\end{example}

\begin{figure}[hbt]
    \centering
    \begin{tikzpicture}
      \node[state, initial] (0) {$0$};
      \node[state, right of=0] (1) {$1$};
      \node[state, right of=1] (2) {$2$};
      \node[state, right of=2] (3) {$3$};
      \node[state, right of=3] (5) {$5$};
      \node[state, right of=5] (6) {$6$};
      \node[state, above of=2] (4) {$4$};
      \node[state, above of=5] (8) {$8$};
      \node[state, below of=3] (7) {$7$};
      \node[state, below of=6] (9) {$9$};
      \node[state, right of=9] (10) {$10$};
      \draw (0) edge node {$a$} (1);
      \draw (1) edge node {$b$} (2);
      \draw (2) edge node {$b$} (3);
      \draw (3) edge node {$a$} (5);
      \draw (5) edge node {$b$} (6);
      \draw (4) edge node {$b$} (3);
      \draw (8) edge node {$b$} (6);
      \draw (4) edge node {$a$} (8);
        \draw (7) edge node {$a$} (9);
        \draw (8) edge[below] node {$a$} (9);
      \draw (9) edge node {$b$} (10);
      \draw (0) edge[bend left] node {$b$} (4);
      \draw (2) edge[below, bend right] node {$a$} (7);
      \draw (1) edge[below, bend right=40] node {$a$} (9);
      \draw (8) edge[bend left=90] node {$a$} (9);
      \tikzset{node distance=0.7cm};
      \node[above of=0] (m0) {$r,t$};
      \node[above of=1] (m1) {$r,t$};
      \node[above of=2] (m2) {$r,t$};
      \node[above of=3] (m3) {$r$};
      \node[above of=4] (m4) {$r,t$};
      \node[above of=5] (m5) {$r$};
      \node[above of=6] (m6) {$r$};
      \node[above of=7] (m7) {$t$};
      \node[above of=8] (m8) {$r,t$};
      \node[above of=9] (m9) {$t$};
      \node[above of=10] (m10) {$t$};
      \draw (4) edge[dashed, bend right=60] node {} (0);
      \draw (1) edge[dashed, bend left] node {} (0);
      \draw (2) edge[dashed] node {} (4);
       \draw (3) edge[dashed, bend right] node {} (4);
       \draw (5) edge[dashed] node {} (8);
     \draw (6) edge[dashed, bend left] node {} (2);
        \draw (7) edge[dashed] node {} (8);
       \draw (8) edge[dashed, bend right=80] node {} (1);
     \draw (9) edge[dashed, bend left=80] node {} (1);
       \draw (10) edge[dashed, bend left=90] node {} (2);
      \end{tikzpicture}
      \caption{Automaton~$\DAWG(P)$ for $P=\{abbab, abaab\}$. Marks $r,t$ above states are defined in Section \ref{section.specific}}
      \label{figure.1}
\end{figure}
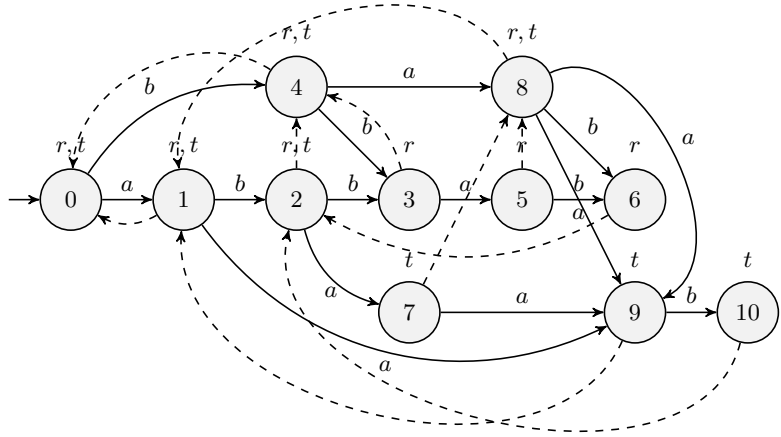
We let $s$ denote the suffix link function associated with~$\DAWG(P)$.
We first define the function $s'$ from $\Fact(P)\setminus\{\varepsilon\}$ to $\Fact(P)$ as follows: for $v \in
\Fact(P)\setminus\{\varepsilon\}$, $s'(v)$ is the longest word~$u \in
\Fact(P)$ that is a suffix of $v$ and for which $u \not
\equiv_{\Suff(P)} v$. Then, if $p=\delta(i,v)$, $s(p)$ is the state $\delta(i,s'(v))$.

\section{Computing the set of {\it T}-specific words}
  \label{section.specific}

In this section, we assume that the reference $R$ and the target $T$
are two finite sets of words and our goal is to compute the set of {\it T}-specific factors of $T$ against $R$. To do so, We first compute the directed acyclic word graph $\DAWG(R \cup T)=(Q, A, i, Q, \delta)$ of $R \cup T$.
Further, we compute a table $\Mark$ indexed by the set of states $Q$ that satisfies: for each state $p$ in $Q$, $\Mark[p]$ is one of the three values $r$, $t$ or both $r,t$ according to the fact that each word labeling a path from $i$ to
$q$ is a factor of some word in $R$ and not of a word in $T$, or is a
factor of a word in $T$ and not of a word in $R$, or is a factor of a word in $R$ and of a word in $T$. This information can be obtained
during the construction of the directed acyclic word graph without increasing the time and space complexity.

The following algorithm outputs a trie (digital tree) of the set of $T$-specific words with respect to $T$ and $R$.

\begin{algo}{Specific-trie}{(Q, A, i, Q, \delta) \mbox{ DAWG of }(R \cup T), s \mbox{ its suffix link}} 
\DOFOR{\mbox{ each } p \in Q \mbox{ with }  \Mark[p] = r,t
  \mbox{ in width-first search from } i}
  \EXT{\textbf{ and for} \mbox{ each } a \in A}
 \IF{(\delta(p,a) \mbox{ defined } \textbf{ and }  \Mark[\delta(p,a)] = t) \textbf{ and } ((p = i) \textbf{ or } }
 \EXT{(\delta(s(p),a) \mbox{ defined }
 \textbf{ and }
 \Mark[\delta(s(p),a)] = r \mbox{ or } r, t))}
\SET{\delta'(p,a)}{\mbox{ new sink}}
\ELSE
\IF{(\delta(p,a) = q \mbox{ with } \Mark[q] = r,t) }
\EXT{\textbf{ and } (q \mbox{ not already reached})}
\SET{\delta'(p,a)}{q}
\FI
\FI
\OD
\RETURN{\A \mbox{, the automaton } (Q, A, i, \{\mbox{sinks}\}, \delta')}
\end{algo}

\begin{example}
The automaton $\DAWG(R \cup T)$
with the input $R=\{abbab\},  T =\{abaab\}$ is shown in Figure~\ref{figure.1}.
 The output of algorithm~{\sc Specific-trie} on
  $\DAWG$ $(R \cup T)$ is shown in
Figure~\ref{figure.2} where the squares are final or sink states of
the trie. The set of $T$-specific words with respect to $R$ is
$\{aa, aba\}$.
\end{example}

\begin{figure}[hbt]
    \centering
    \begin{tikzpicture}
      \node[state, initial] (0) {$0$};
      \node[state, right of=0] (1) {$1$};
      \node[state, right of=1] (2) {$2$};
      \node[state, above of=2] (4) {$4$};
      \node[state, above of=5] (8) {$8$};
      \draw (0) edge node {$a$} (1);
      \draw (1) edge node {$b$} (2);
      \draw (4) edge node {$a$} (8);
      \draw (0) edge[bend left] node {$b$} (4);
      \tikzset{node distance=0.6cm};
      \node[below of=0] (m0) {$r,t$};
      \node[below of=1] (m1) {$r,t$};
      \node[below of=2] (m2) {$r,t$};
      \node[below of=4] (m4) {$r,t$};
      \node[below of=8] (m8) {$r,t$};
      \tikzset{node distance=2cm};
      \tikzset{every state/.append style={rectangle}}
      \tikzset{every state/.style={ 
          {rectangle},minimum size=15pt,
                  semithick,
                  fill=black}};
      \node[state, below of=2] (s2) {$$};
       \node[state, below of=1] (s1) {$$};
       \draw (1) edge node {$a$} (s1);
       \draw (2) edge node {$a$} (s2);
      \end{tikzpicture}
      \caption{The trie of $T$-specific words with respect to $R$.}
      \label{figure.2}
        \end{figure}
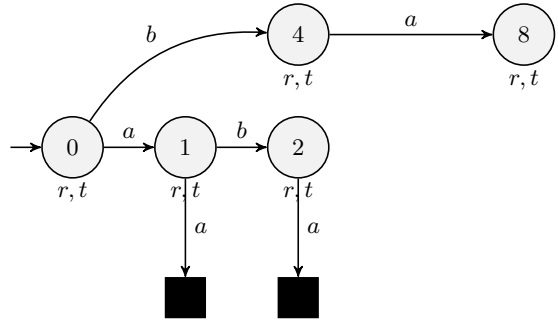

\begin{proposition}
  Let $\DAWG(R \cup T)$ be the output of algorithm {\sc Dawg} on the finite set of
  words~$R \cup T$, let $s$ be its suffix function, and let $\Mark$ be the table defined as above. Algorithm {\sc
  Specific-trie} builds the trie recognizing the set of
$T$-specific words with respect to $R$.
\end{proposition}

\begin{proof}
Let $S$ be the set of of
$T$-specific words with respect to $R$.

Consider a word $ua$ ($a \in A$) accepted by $\A$. Note that $\A$ accepts only nonempty words. Let $p = \delta'(i, u)$. Since the $\DAWG$ automaton is processed with a width-first search, $u$ is the shortest word for which $\delta(i, u) = p$. Therefore, if $u = bv$ with $b \in A$, we have $\delta(i, v) = s(p)$ by definition of the suffix function $s$. When the test “($\delta(p, a)$ defined and $\Mark[\delta(p, a)] = t$) and ($\delta(s(p), a)$ defined and $\Mark[\delta(s(p), a)] = r$ or $r,t$)” is satisfied, this implies that $va \in \Fact(R)$. Thus, $bva \notin \Fact(R)$, while
$bv, va \in \Fact(R)$ and $bva \in \Fact(T)$. So, $ua$ is 
a $T$-specific word with respect to $R$.
If $u$ is the empty word, then $p = i$. 
The transition from $i$
to the sink labeled by $a$ is created under the condition
“$\delta(p, a)$ defined and $\Mark[\delta(p, a)] = t$”, 
which means that $a \in \Fact(T)$. The word $a$ is again 
a $T$-specific word with respect to $R$.
Thus the words accepted by $\A$ are $T$-specific words with respect to $R$.

Conversely, let $ua \in S$. If $u$ is the empty word, this means that $a$ does not occur in $\Fact(R)$ and occurs in $\Fact(T)$ therefore there is a transition labeled by $a$ from $i$ in $\DAWG(R \cup T)$ to a state marked $t$. Thus a transition from $i$ to a sink state in $\A$ created Line 3 and $a$ is accepted by $\A$.
Now assume that $u = bv$. The word $u$ is in $\Fact(R)$. So let $p = \delta(i, u)$. Note that $u$ is
the shortest word for which $p = \delta(i, u)$, because all
such words are suffixes of each other in the $\DAWG$ automaton. 
The word $ua$ is not in $\Fact(R)$ and is in $\Fact(T)$, so
the condition “$\delta(p, a)$ defined and $\Mark[p, a] = t$” is satisfied. Let $q = s(p)$. We have $q = \delta(i, v)$ 
because of the minimality of the length of $u$ and the definition of $s$. 
Since $va$ is in $\Fact(R)$, the condition
“$\delta(s(p), a)$ defined and $\Mark[\delta(s(p),a)] = r$ or $r,t$” at Line 2 is satisfied which yields
the creation of a transition at Line 3 to make $\A$ accept
$ua$ as wanted. 
\end{proof}

A main point in algorithm {\sc Specific-trie} is that it uses the
function $s$ defined on states of the input DAWG. It is not
possible to proceed similarly when considering the minimal
factor automaton of~$\Fact(R \cup T)$ because there is no analogue
function $s$. However, it is possible to reduce the automaton
$\DAWG(R \cup T)$ by merging states having the same future (right
context) and the same image by $s$. For example, on the DAWG of
Figure~\ref{figure.1}, states $6$ and $10$ can be merged because
$s(6)=s(10)=2$. States $3$ and $7$, nor states $5$ and $9$ cannot
be merged with the same argument.

\begin{proposition}
  Algorithms {\sc Dawg} and {\sc Specific-trie} together run in time\\
  $O(\size(R \cup T)\times |A|)$ with input two finite sets of words
  $R, T$, if the transition
  functions are implemented by transition matrices.
\end{proposition}

If $P$ is a set of words, we denote by $A_P$ the set of letters occurring in $P$.

\begin{proposition}
  Let $R, T$ be two finite sets of words.
  The number of $T$-specific words with respect to $R$ is no more than
   $(2 \size(R) - 2)(|A_R|-1) +
  |A_T \setminus A_R|-|A_R| +m$, if $\size(R) > 1$, where~$m$ the
  number of words in $R$.
  The bound becomes $|A_T \setminus A_R|$ when $\size(R) \leq 1$.
\end{proposition}
\begin{proof}
We let $S$ denote the set of $T$-specific words with respect to $R$. Since $S$
is included in the set of minimal forbidden words of $\Fact(R)$ with respect to the
alphabet $A = A_R \cup A_T$, the bound comes from \cite[Corollary
4.1]{DBLP:journals/fuin/BealCM03}.
\end{proof}

\def\dd{\mathinner{\ldotp\ldotp}}	
\def\Ts{\mathit{Ts}}

\section{Computing occurrences of target-specific factors: the {\it T}-specific table}
In this section, we consider that $R$ and $T$ are just words.
The goal of the section is to design an algorithm that computes all the
occurrences of $T$-specific words in $T$. To do so, we define the
$T$-specific table associated with the pair $R,T$ of words of the problem. 

A letter of $T$ at position $k$ is denoted by $T[k]$ and $T[i\dd j]$ denotes the factor $T[i]T[i+1]\cdots T[j]$ of $T$. Then, the {\it T}-specific table $\mathit{Ts}$ is defined, for $i=0,\dots,|T|-1$, by

$$\Ts[i] = \left\{
\begin{array}{ll}
 j, & \mbox{ if } T[i\dd j] \mbox{ is } T\mbox{-specific}, i\leq j,\\
-1, & \mbox{ else.}
\end{array}
\right.$$

\paragraph{Note 1.}
Since the set of $T$-specific factors is both prefix-free and suffix-free,
for each position $k$ on $T$ there is at most one $T$-specific factor of $T$
starting at $k$ and for each position $j$ on $T$ there is at most one
$T$-specific factor of $T$ ending at $j$.

\paragraph{Note 2.}
Instead of computing the {\it T}-specific table $\mathit{Ts}$, in a straightforward way, the algorithm below can be transformed to compute the list of pairs $(i,j)$ of positions on $T$ for which $\Ts[i] =j$ and $j\neq -1$.

\medskip
To compute the table we use $\mathcal{R}$, the suffix automaton of $R$, with
its transition function $\delta$ and equipped with both the suffix link $s$
(used here as a failure link) and the length function $\ell$ defined on states by: $\ell[p]=\max\{z\in A^*\mid \delta(i,z)=p\}$. Functions $s$ and $\ell$ transform the automaton into a search machine, see \cite[Section 6.6]{CHL07cup}.

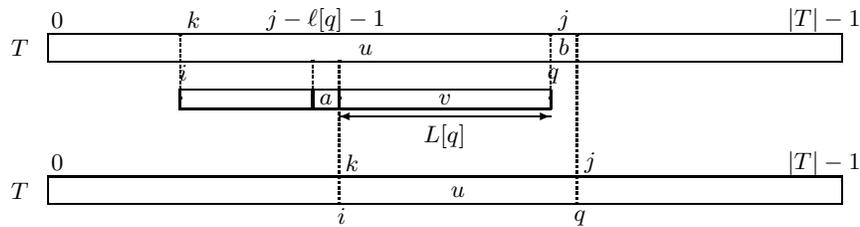
\begin{figure}[hbt]
{\begin{picture}(300,90)(-15,-40)
\put(-14,31){$T$}
\put(  0,29){\framebox(300,10){}}
\put(  0,41){\makebox(7,7){$0$}}
\put( 50,41){\makebox(10,7){$k$}}
\put(100,41){\makebox(10,7){$j-\ell[q]-1$}}
\put(190,41){\makebox(10,7){$j$}}
\put(290,41){\makebox(7,7){$|T|-1$}}
\put( 50,30){\makebox(140,7){$u$}}
\put(190,31){\makebox(10,7){$b$}}
\put( 46,21){\makebox(10,7){$i$}}
\put(186,20){\makebox(10,7){$q$}}
\put(110,8){\vector(1,0){80}}
\put(190,8){\vector(-1,0){80}}
\put(110,-3){\makebox(80,7){$L[q]$}}
\put( 50,11){\framebox(50,7){}}
\put(100,11){\framebox(10,7){$a$}}
\put(110,11){\framebox(80,7){$v$}}
\put( 50,15){\dashbox{1}(0,24){}}
\put(110,15){\dashbox{1}(0,14){}}
\put(100,15){\dashbox{1}(0,14){}}
\put(110,15){\dashbox{1}(0,14){}}
\put(190,15){\dashbox{1}(0,24){}}
\put(200,29){\dashbox{1}(0,10){}}
\put(-14,-23){$T$}
\put(  0,-25){\framebox(300,10){}}
\put(  0,-13){\makebox(7,7){$0$}}
\put(110,-13){\makebox(10,7){$k$}}
\put(200,-13){\makebox(10,7){$j$}}
\put(290,-13){\makebox(7,7){$|T|-1$}}
\put(110,-24){\makebox(90,7){$u$}}

\put(106,-33){\makebox(10,7){$i$}}
\put(196,-34){\makebox(10,7){$q$}}

\put(110,-25){\dashbox{1}(0,36){}}
\put(200,-25){\dashbox{1}(0,54){}}
\end{picture}}
\caption{A $T$-specific word found: when $u\in\Fact(R)$ and
$ub\not\in\Fact(R)$, either $avb$ or $b$ is a $T$-specific factor with respect to $R$ ($a,b$ are letters). \label{fig6.1}}
\end{figure}

Figure~\ref{fig6.1} illustrates the principle of Algorithm \Algo{TsTable}. Let us assume the factor $u=T[k\dd j-1]$ is a factor of $R$ but $ub$ is not for some letter $b$.
Then, let $v$ be the longest suffix of $u$ for which $vb$ is a factor of $R$.
If it exists, then clearly $avb$, with $a$ letter preceding $v$, is
$T$-specific. Indeed, $av, vb\in\Fact(R)$ and $avb\not\in\Fact(R)$, which
means that $avb$ is a minimal forbidden word of $R$ while occurring in $T$.
Therefore, setting $q=\delta(i,u)$, $\Ts[j-\ell[q]-1]=j$ since $\ell[q]=|v|$ due to a property of the DAWG of $\mathcal{R}$. If there is no suffix of $u$ satisfying the condition, the letter $b$ alone is $T$-specific and $\Ts[j]=j$.

\medskip
\begin{algo}{TsTable}{T \mbox{ target word},\mathcal{R}
 \mbox{ DAWG}(R), i \textit{ initial}(\mathcal{R})}
  \SET{(q,j)}{(i,0)} \label{algo2-1}
  \DOWHILE{j < |T|}
    \SET{\Ts[j]}{-1}
    \IF{\delta(q,T[j]) \mbox{ undefined}}
      \DOWHILE{q \neq i
          \mbox{ and } \delta(q,T[j]) \mbox{ undefined}} \label{algo2-5}
        \SET{q}{s[q]}
      \OD \label{algo2-6}
      \IF{\delta(q,T[j]) \mbox{ undefined}} \RCOM{8}{$q = i$} \label{algo2-7}
        \SET{\Ts[j]}{j}
        \SET{j}{j+1}
      \ELSE
        \SET{\Ts[j-\ell[q]-1]}{j}
        \SET{(q,j)}{(\delta(q,T[j]),j+1)}
      \FI
    \ELSE
      \SET{(q,j)}{(\delta(q,T[j]),j+1)} \label{algo2-12}
    \FI
  \OD
\RETURN{\Ts}
\end{algo}

\begin{theorem}
The DAWG of the reference set $R$ of words being preprocessed, applied to
a word $T$, Algorithm \Algo{TsTable} computes its $T$-specific table with
respect to $R$ and runs in linear time, i.e. $O(|T|)$ on a fixed-size
alphabet.
\end{theorem}

\begin{proof}
The algorithm implements the ideas detailed above. A more formal proof relies
on the invariant of the while loop: $q=\delta(i,u)$, where $i$ is the initial
state of the suffix automaton of $R$ and $u=T[k\dd j]$ for
a position $k\leq j$. Since $k=j-|u|$, it is left implicit in the algorithm.
The length $|u|$ could be computed and then incremented when $j$ is. It is
made explicit only at line 10 as $L[q]+1$ after computing the suffix $v$ of
$u$.

For example, when $v$ exists, $u$ is changed to $vb$ and $j$
is incremented, which maintains the equality.

As for the running time, note that instructions at lines~\ref{algo2-1} and \ref{algo2-7}-\ref{algo2-12} execute
in constant time for each value of $j$. All the executions of the instruction
at line~\ref{algo2-6} execute in time $O(|T|)$ because the link $s$ reduces strictly the
potential length of the $T$-specific word ending at $j$, that is, it
virtually increments the starting position of $v$ in the picture.

Thus the whole execution is done in time $O(|T|)$.
\end{proof}

Algorithm \Algo{TsTable} can be improved to run in real-time on a fixed-size alphabet. This is done by optimizing the suffix link $s$ defined on the automaton $\mathcal{R}$. To do so, let us define, for each state $q$ of $\mathcal{R}$,
$$\out(q)=\{a \mid \delta(q,a) \mbox{ defined for letter } a\}.$$
Then, the optimised suffix link $G$ is defined by $G[\textit{initial}(\mathcal{R})]=\textbf{nil}$ and, for any other state $q$ of $\mathcal{R}$, by
$$G[q]=\left\{\begin{array}{ll}
 s[q], & \mbox{ if } \out(q)\subset\out(s[q]),\\
 G[s[q]], & \mbox{ else.} 
\end{array}
\right.$$
Note that, since we always have $\out(q)\subseteq\out(s[q])$, the definition of $G$ can be reformulated as
$$G[q]=\left\{\begin{array}{ll}
 s[q], & \mbox{ if } \deg(q)<\deg(s[q]),\\
 G[s[q]], & \mbox{ else,} 
\end{array}
\right.$$
where $\deg$ is the outgoing degree of a state. 
Therefore, its computation can be realized in linear time with respect to the number of states of $\mathcal{R}$. After substituting $G$ for $s$ in Algorithm \Algo{TsTable}, when the alphabet is of size $\alpha$ the instruction at line~\ref{algo2-6} executes no more than $\alpha$ times for each value of $q$. So the time to process a given state $q$ is constant.
This is summarized in the next corollary.

\begin{corollary}
When using the optimized suffix link, Algorithm \Algo{TsTable} runs in real-time on a fixed-size alphabet.
\end{corollary}

On a more general alphabet of size $\alpha$, the processing of a given state of the automaton can be done in time $\log\alpha$.



\bibliographystyle{abbrv}
\bibliography{algo}

\end{document}